\documentclass[runningheads]{llncs}
\usepackage{graphicx}

\usepackage{amssymb, bussproofs, amsmath}
\usepackage[all]{xy}
\usepackage{pgf, tikz, color, changepage, relsize, tcolorbox}
\usetikzlibrary{decorations.pathreplacing}
\usepackage[utf8]{inputenc}
\usepackage{enumerate}
\usepackage{rotating}
\usepackage{appendix}

\newcommand{\tstitl}{\mathsf{G3Tstit}}

\newcommand{\xstitl}{\mathsf{G3Xstit}}

\newcommand{\ldm}{\mathsf{Ldm}}
\newcommand{\ldmn}{\mathsf{Ldm}}

\newcommand{\ldmt}{\mathsf{Tstit}}
\newcommand{\tstit}{\mathsf{Tstit}}
\newcommand{\xstit}{\mathsf{Xstit}}

\newcommand{\settrefl}{(\mathsf{refl}_{\Box})}
\newcommand{\setteucl}{(\mathsf{eucl}_{\Box})}

\newcommand{\ioa}{(\mathsf{IOA})}
\newcommand{\bridge}{(\mathsf{br}_{[i]})}
\newcommand{\gtrans}{(\mathsf{trans}_{\g})}
\newcommand{\gser}{(\mathsf{ser}_{\g})}
\newcommand{\gconn}{(\mathsf{conn}_{\g})}
\newcommand{\hconn}{(\mathsf{conn}_{\h})}
\newcommand{\gconv}{(\mathsf{conv}_{\g})}
\newcommand{\hconv}{(\mathsf{conv}_{\h})}
\newcommand{\ncuh}{(\mathsf{ncuh})}

\newcommand{\agtd}{(\mathsf{agd})}
\newcommand{\irrtwo}{(\mathsf{irr}_{\g})}

\newcommand{\comp}{(\mathsf{comp}_{\g 1})}
\newcommand{\comptwo}{(\mathsf{comp}_{\g 2})}

\newcommand{\ioax}{(\mathsf{IOA_X})}

\newcommand{\id}{(\mathsf{id})}

\newcommand{\R}{\mathcal{R}}

\newcommand{\lb}{\langle}
\newcommand{\rb}{\rangle}
\newcommand{\g}{\mathsf{G}}
\newcommand{\h}{\mathsf{H}}

\newcommand{\pres}{\mathsf{pres}}
\newcommand{\fut}{\mathsf{fut}}

\EnableBpAbbreviations
\begin{document}
%
%
%
\title{Appendix for: Cut-free Calculi and Relational Semantics for Temporal STIT logics}
\author{Kees van Berkel\inst{1} \and Tim Lyon\inst{1}}

\authorrunning{Berkel and Lyon}

\institute{Institut f\"ur Logic and Computation, Technische Universit\"at Wien, 1040 Wien, Austria  \\ \email{\{kees,lyon\}@logic.at}}
%
%
%
%
\maketitle              
\begin{abstract}
This paper is an appendix to the paper ``Cut-free Calculi and Relational Semantics for Temporal STIT logics'' by Berkel and Lyon, 2019 \cite{BerLyo19}. It provides the completeness proof for the basic STIT logic $\ldm$ (relative to irreflexive, temporal Kripke STIT frames) as well as gives the derivation of the independence of agents axiom for the logic $\xstit$.
\end{abstract}


\appendix

   \section{Completeness of $\ldm$}\label{Completeness_Proof_ldm}
We give the definitions and lemmas sufficient to prove the completeness of $\ldm$ relative to $\tstit$ frames \cite{Lor13,BerLyo19}. We make use of the canonical model of $\ldm$ (obtained by standard means \cite{BlaRijVen01,BalHerTro08}) to construct a $\tstit$ model. A truth-lemma is then given relative to this model, from which, completeness follows as a corollary.

\begin{definition}[$\ldmn$-CS, $\ldmn$-MCS] A set $\Theta \subset \mathcal{L}_{\ldmn}$ is a \emph{$\ldmn$ consistent set ($\ldmn$-CS)} iff $\Theta \not\vdash_{\ldmn} \bot$. We call a set $\Theta \subset \mathcal{L}_{\ldmn}$ a \emph{$\ldmn$ maximally consistent set ($\ldmn$-MCS)} iff $\Theta$ is a $\ldmn$-CS and for any set $\Theta'$ such that $\Theta \subset \Theta'$, $\Theta' \vdash_{\ldmn} \bot$.
\end{definition}

\begin{lemma}[Lindenbaum's Lemma \cite{BlaRijVen01}]\label{Lindenbaum}
Every $\ldmn$-CS can be extended to a $\ldmn$-MCS.
\end{lemma}

\begin{definition}[Present and Future Pre-Canonical $\tstit$ Model] The \emph{present pre-canonical $\tstit$ model} is the tuple $M^{\pres} = (W^{\pres}, \R^{\pres}_{\Box},$ $\{\R^{\pres}_{i}|i \in Ag\}, V^{\pres})$ defined below left, and the \emph{future pre-canonical $\tstit$ model} is the tuple $M^{\fut} = (W^{\fut}, \R^{\fut}_{\Box}, \{\R^{\fut}_{i}|i \in Ag\}, V^{\fut})$ defined below right:

\begin{center}
\begin{multicols}{2}

\begin{center}

\begin{itemize}

\item $W^{\pres}$ is the set of all $\ldm$-MCSs;

\item $\R^{\pres}_{\Box}wu$ iff for all $\Box \phi \in w$, $\phi \in u$;

\item $\R^{\pres}_{i}wu$ iff for all $[i] \phi \in w$, $\phi \in u$;

\item $V^{\pres}(p) = \{w \in W| p \in w\}$.

\end{itemize}
\end{center}

\begin{center}

\begin{itemize}

\item $W^{\fut} = W^{\pres}$;

\item $\R^{\fut}_{\Box}(w) = \bigcap_{i \in Ag} \R^{\pres}_{i}(w)$;

\item $\R^{\fut}_{i}(w) = \bigcap_{i \in Ag} \R^{\pres}_{i}(w)$;

\item $V^{\fut}(p) = V^{\pres}(p)$.

\end{itemize}
\end{center}

\end{multicols}
\end{center}

\end{definition}

\begin{definition}[Canonical Temporal Kripke STIT Model] We define the \emph{canonical temporal Kripke STIT model} to be the tuple $M^{\ldm} = (W^{\ldm}, \R^{\ldm}_{\Box}, $ $\{R^{\ldm}_{i} | i \in Ag\}, \R^{\ldm}_{Ag}, \R^{\ldm}_{\g}, \R^{\ldm}_{\h},$ $V^{\ldm})$ such that:

\begin{itemize}

\item $W^{\ldm} = W^{\pres} \times \mathbb{N}$\footnote{Note that we choose to write each world $(w,j) \in W^{\ldm}$ as $w^{j}$ to simplify notation. Moreover, we write $\phi \in w^{j}$ to mean that the formula $\phi$ is in the $\ldm$-MCS $w$ associated with $j$.};

\item $\R^{\ldm}_{\Box}w^{j}u^{j}$ iff (i) $\R^{\pres}_{\Box}wu$ and $j = 0$, or (ii) $\R^{\fut}_{\Box}wu$ and $j > 0$;

\item $\R^{\ldm}_{i}w^{j}u^{j}$ iff (i) $\R^{\pres}_{i}wu$ and $j = 0$, or (ii) $\R^{\fut}_{i}wu$ and $j > 0$;

\item $\R^{\ldm}_{Ag}(w^{j}) = \bigcap_{1 \leq i \leq n} \R^{\ldm}_{i}(w^{j})$;

\item $\R^{\ldm}_{\g} = \{(w^{j},w^{k})| w^{j}, w^{k} \in W^{\ldm} \text{ and } j < k \}$;

\item $\R^{\ldm}_{\h} = \{(u^{i},w^{i})| (w^{i},u^{i}) \in \R^{\ldm}_{\g}\}$;

\item $V^{\ldm}(p) = \{w^{j} \in W^{\ldm}| w \in V^{\pres}(p)\}$.

\end{itemize}

\end{definition}

\begin{lemma}\label{Relate_only_same_level} For all $\alpha \in \{\Box, Ag\} \cup Ag$, if $\R^{\ldm}_{\alpha}w^{j}u^{k}$ for $j,k \in \mathbb{N}$, then $j = k$.
\end{lemma}

\begin{proof} Follows by definition of the canonical $\tstit$ model.

\end{proof}

\begin{lemma}\label{AG_Rel_Same_All_Levels} For all $j \in \mathbb{N}$ with $k \geq 1$, $(w^{j},u^{j}) \in \R^{\ldm}_{Ag}$ iff $(w^{j+k},u^{j+k}) \in \R^{\ldm}_{Ag}$.
\end{lemma}

\begin{proof} This follows from the fact that $u^{0} \in \R^{\ldm}_{Ag}(w^{0})$ iff $u \in \bigcap_{i \in Ag} \R^{\pres}_{i}(w)$ iff $u \in \R^{\fut}_{i}(w)$ for each $i \in Ag$ iff $u \in \bigcap_{i \in Ag} \R^{\fut}_{i}(w)$ iff $u^{k} \in \bigcap_{i \in Ag} \R^{\ldm}_{i}(w^{k})$ for any $k > 0$.
\end{proof}

\begin{lemma}[\cite{BlaRijVen01}]\label{Relations_Diamonds} (i) For all $\mathsf{x} \in\{\pres, \fut, \ldm\}$, $\R^{\mathsf{x}}_{\Box}wu$ iff for all $\phi$, if $\phi \in u$, then $\Diamond \phi \in w$. (ii) For all $\mathsf{x} \in\{\pres, \fut, \ldm\}$, $\R^{\mathsf{x}}_{i}wu$ iff for all $\phi$, if $\phi \in u$, then $\lb i \rb \phi \in w$.

\end{lemma}


\begin{lemma}[Existence Lemma \cite{BlaRijVen01}]\label{Existence_Lemma} (i) For any world $w^{j} \in W^{\ldm}$, if $\Diamond \phi \in w^{j}$, then there exists a world $u^{j} \in W^{\ldm}$ such that $\R^{\ldm}_{\Box}w^{j}u^{j}$ and $\phi \in u^{j}$. (ii) For any world $w^{j} \in W^{\ldm}$, if $\lb i \rb \phi \in w^{j}$, then there exists a world $u^{j} \in W^{\ldm}$ such that $\R^{\ldm}_{i}w^{j}u^{j}$ and $\phi \in u^{j}$.

\end{lemma}


\begin{lemma}\label{Canonical_is_TTKSTIT_Model}
The Canonical Model is a temporal Kripke STIT model.
\end{lemma}

\begin{proof} We prove that $M^{\ldm}$ has all the properties of a $\tstit$ model:

\begin{itemize}

\item By lemma \ref{Lindenbaum}, the $\ldm$ consistent set $\{p\}$ can be extended to a $\ldm$-MCS, and therefore $W^{\pres}$ is non-empty. Since $\mathbb{N}$ is non-empty as well, $W^{\pres} \times \mathbb{N} = W^{\ldm}$ is a non-empty set of worlds.

\item We argue that $\R^{\ldm}_{\Box}$ is an equivalence relation between worlds of $W^{\ldm}$, and omit the arguments for $\R^{\ldm}_{i}$ and $\R^{\ldm}_{Ag}$, which are similar. Suppose that $w^{j} \in W^{\ldm}$. We have two cases to consider: (i) $j = 0$, and (ii) $j > 0$. \textbf{(i)} Standard canonical model arguments apply and $\R^{\ldm}_{\Box}$ is an equivalence relation between all worlds of the form $w^{0} \in W^{\ldm}$ (See \cite{BlaRijVen01} for details). \textbf{(ii)} If we fix a $j >0$, then $\R^{\ldm}_{\Box}$ will be an equivalence relation for all worlds of the form $w^{j} \in W^{\ldm}$ since the intersection of equivalence relations produces another equivalence relation. Last, since $\R^{\ldm}_{\Box}$ is an equivalence relation for each fixed $j \in \mathbb{N}$, and because each $W^{\pres} \times \{j\} \subset W^{\ldm}$ is disjoint from each $W^{\pres} \times \{j'\} \subset W^{\ldm}$ for $j \neq j'$, we know that the union all such equivalence relations will be an equivalence relation.

\item[{\rm \textbf{(C1)}}] Let $i$ be in $Ag$ and assume that $(w^{j},u^{j}) \in \R^{\ldm}_{i}$. We split the proof into two cases: (i) $j = 0$, or (ii) $j > 0$. \textbf{(i)} Assume that $\Box \phi \in w^{0}$. Since $w$ is a $\ldm$-MCS, it contains the axiom $\Box \phi \rightarrow [i] \phi$, and so, $[i] \phi \in w$ as well. Since $(w,u) \in \R^{\pres}_{i}$ (because $j = 0$), we know that $\phi \in u$ by the definition of the relation; therefore, $(w,u) \in \R^{\pres}_{\Box}$, which implies that $(w^{0},u^{0}) \in \R^{\ldm}_{\Box}$ by definition. \textbf{(ii)} The assumption that $j >0$ implies that $u \in \R^{\fut}_{i} (w) $ $= \bigcap_{i \in Ag} \R^{\pres}_{i}(w) = \R^{\fut}_{\Box}(w)$ by definition, which implies that $(w^{j},u^{j}) \in \R^{\ldm}_{\Box}$.

\item[{\rm \textbf{(C2)}}] Let $u^{j}_{1}, ..., u^{j}_{n} \in W^{\ldm}$ and assume that $\R^{\ldm}_{\Box}u^{j}_{i}u^{j}_{k}$ for all $i,k \in \{1,...,n\}$. We split the proof into two cases: (i) $j = 0$, or (ii) $j > 0$. \textbf{(i)} We want to show that there exists a world $w^{j} \in W^{\ldm}$ such that $w^{j} \in \bigcap_{1 \leq i \leq n} \R^{\ldm}_{i}(u^{j}_{i})$. Let $\hat{w}^{j} = \bigcup_{1 \leq i \leq n} \{\phi | [i] \phi \in u^{j}_{i} \}$. Suppose that $\hat{w}^{j}$ is inconsistent to derive a contradiction. Then, there are $\psi_{1}$,...,$\psi_{k}$ such that $\vdash_{\ldm} \bigwedge_{1 \leq l \leq k} \psi_{i} \rightarrow \bot$. For each $i \in Ag$, we define $\Phi_{i} = \{\psi_{l}|[i] \psi_{l} \in u^{j}_{i}\} \subseteq \{\psi_{1},...,\psi_{k}\}$. Observe that for each $i \in Ag$, $[i] \bigwedge \Phi_{i} \in u^{j}_{i}$ because $\bigwedge [i] \Phi_{i} \in u^{j}_{i}$ and $\vdash_{\ldm} \bigwedge [i] \Phi_{i} \rightarrow [i] \bigwedge \Phi_{i}$. Since by assumption $\R^{\ldm}_{\Box}u^{j}_{i}u^{j}_{k}$ for all $i,k \in \{1,...,n\}$, this means that for any $u^{j}_{m}$ we pick (with $1 \leq m \leq n$), $\Diamond [i] \bigwedge \Phi_{i} \in u^{j}_{m}$ for each $i \in Ag$ by lemma \ref{Relations_Diamonds}; hence, $\bigwedge_{i \in Ag} \Diamond [i] \bigwedge \Phi_{i} \in u^{j}_{m}$. By the $\ioa$ axiom, this implies that $\Diamond \bigwedge_{i \in Ag} [i] (\bigwedge \Phi_{i}) \in u^{j}_{m}$. By lemma \ref{Existence_Lemma}, there must exist a world $v^{j}$ such that $\R^{\ldm}_{\Box}u^{j}_{m}v^{j}$ and $\bigwedge_{i \in Ag} [i] (\bigwedge \Phi_{i}) \in v^{j}$. But then, since $\vdash_{\ldm} [i] (\bigwedge \Phi_{i}) \rightarrow \bigwedge \Phi_{i}$ by reflexivity, $\vdash_{\ldm} \bigwedge_{i \in Ag} (\bigwedge \Phi_{i}) \leftrightarrow \bigwedge_{1 \leq i \leq k} \psi_{i}$, and $\vdash_{\ldm} \bigwedge_{1 \leq i \leq k} \psi_{i} \rightarrow \bot$, it follows that $\bot \in v^{j}$, which is a contradiction since $v^{j}$ is a $\ldm$-MCS. Therefore, $\hat{w}^{j}$ must be consistent and by lemma \ref{Lindenbaum}, it may be extended to a $\ldm$-MCS $w^{j}$. Since for each $[i] \phi \in u^{j}_{i}$, $\phi \in w^{j}$, we have that $w \in \R^{\pres}_{i}(u_{i})$ for each $i \in Ag$. Hence, $w \in \bigcap_{1 \leq i \leq n} \R^{\pres}_{i}(u_{i})$, and so, $w^{j} \in \bigcap_{1 \leq i \leq n} \R^{\ldm}_{i}(u^{j}_{i})$. \textbf{(ii)} Suppose that $j > 0$, so that $t^{j} \in \R^{\ldm}_{\Box}(s^{j})$ iff $t \in \R^{\fut}_{\Box}(s)$ = $\bigcap_{i \in Ag} \R^{\pres}_{i}(s)$. By assumption then, $u^{j}_{m} \in \bigcap_{i \in Ag} \R^{\pres}_{i}(u^{j}_{k}) = \R^{\fut}_{i}(u^{j}_{k})$ for all $k,m \in \{1,...,n\}$ and each $i \in Ag$. Hence, $u^{j}_{m} \in \bigcap_{i \in Ag} \R^{\fut}_{i}(u^{j}_{k})$ for all $k,m \in \{1,...,n\}$. If we therefore pick any $u^{j}_{k}$, it follows that $u^{j}_{k} \in \bigcap_{i \in Ag}\R^{\fut}_{i} (u^{j}_{i})$, meaning that the intersection $\bigcap_{1 \leq i \leq n} \R^{\ldm}_{i}(u^{j}_{i})$ is non-empty.

\item[{\rm \textbf{(C3)}}] Follows by definition.

\item $\R^{\ldm}_{\g}$ is a transitive and serial by definition, and $\R^{\ldm}_{\h}$ is the converse of $\R^{\ldm}_{\g}$ by definition as well.


\item[{\rm \textbf{(C4)}}] For all $u^{j}, u^{k}, u^{l} \in W^{\ldm}$, suppose that $\R^{\ldm}_{\g}u^{j}u^{k}$ and $\R^{\ldm}_{\g}u^{j}u^{l}$.  Then, $j < k$ and $j < l$, and since $\mathbb{N}$ is linearly ordered, we have that $k < l$, $k = l$, or $k> l$, implying that $\R^{\ldm}_{\g}u^{k}u^{l}$, $u^{k} = u^{l}$, or $\R^{\ldm}_{\g}u^{l}u^{k}$.

\item[{\rm \textbf{(C5)}}] Similar to previous case.

\item[{\rm \textbf{(C6)}}] Suppose that $(u^{j},v^{j+k}) \in \R^{\ldm}_{\g} \circ \R^{\ldm}_{\Box}$ with $k \geq 1$. By definition of $\R^{\ldm}_{\g}$, $u^{j+k}$ is the only element in $\R^{\ldm}_{\g}(u^{j})$ associated with $j+k$, and so, $(u^{j+k},v^{j+k}) \in \R^{\ldm}_{\Box}$ (By lemma \ref{Relate_only_same_level} no other $u^{j+k'}$ with $k' \neq k$ can relate to $v^{j+k}$ in $\R^{\ldm}_{\Box}$.). Since $k \geq 1$, $v^{j+k} \in \R^{\ldm}_{\Box}(u^{j+k})$ iff $v \in \R^{\fut}_{\Box}(u) = \bigcap_{i \in Ag} \R^{\pres}_{i}(u)$ iff $v^{0} \in \R^{\ldm}_{Ag}(u^{0})$. By lemma \ref{AG_Rel_Same_All_Levels}, $(u^{j},v^{j}) \in \R^{\ldm}_{Ag}$. This implies that, and since $(v^{j},v^{j+k}) \in \R^{\ldm}_{\g}$ by definition, we have that $(u^{j},v^{j+k}) \in \R^{\ldm}_{Ag} \circ \R^{\ldm}_{\g}$.

\item[{\rm \textbf{(C7)}}] Follows from the definition of the $\R^{\ldm}_{\g}$ relation.



\item Last, it is easy to see that the valuation function $V^{\ldm}$ is indeed a valuation function.

\end{itemize}

\end{proof}

\begin{lemma}[Truth-Lemma]\label{Truth_Lemma} For any formula $\phi$, $M^{\ldm}, w^{0} \models \phi$ iff $\phi \in w^{0}$.

\end{lemma}

\begin{proof} Shown by induction on the complexity of $\phi$ (See \cite{BlaRijVen01}).

\end{proof}

\section{$\xstitl$ Derivation of IOA$^x$ Axiom}\label{XSTIT_IOA}

We make use of the system of rules $\ioax$, to derive the $\xstit$ IOA axiom in $\xstitl$. \\

\begin{small}
\begin{tabular}{@{\hskip -4em} c}
\AxiomC{$R_{\Box}w_{1}w_{2}, R_{\Box}w_{1}w_{3}, R_{\Box}w_{1}w_{4}, R_{A}w_{4}w_{5}, R_{A}w_{2}w_{5}, w_{2}: \langle A \rangle^x \overline{\phi}, w_{3}: \langle B \rangle^x \overline{\psi}, ...\ \ $
$ w_{5}: \phi, w_{5}: \overline{\phi}$}
\UnaryInfC{$R_{\Box}w_{1}w_{2}, R_{\Box}w_{1}w_{3}, R_{\Box}w_{1}w_{4}, R_{A}w_{4}w_{5}, R_{A}w_{2}w_{5}, w_{2}: \langle A \rangle^x \overline{\phi}, w_{3}: \langle B \rangle^x \overline{\psi}, ...\ \ $
$ w_{5}: \phi$}
\RightLabel{($\mathsf{IOA-U_{1}})$}
\UnaryInfC{$R_{\Box}w_{1}w_{2}, R_{\Box}w_{1}w_{3}, R_{\Box}w_{1}w_{4}, R_{A}w_{4}w_{5}, w_{2}: \langle A \rangle^x \overline{\phi}, w_{3}: \langle B \rangle^x \overline{\psi}, ... \ \ $
$ w_{5}: \phi$}
\UnaryInfC{$R_{\Box}w_{1}w_{2}, R_{\Box}w_{1}w_{3}, R_{\Box}w_{1}w_{4}, w_{2}: \langle A \rangle^x \overline{\phi}, w_{3}: \langle B \rangle^x \overline{\psi}, ...\ \ $
$ w_{4}: [A]^x\phi$}
\UnaryInfC{$D_1$}
\DisplayProof
\end{tabular}

\ \\

\begin{tabular}{@{\hskip -4em} c}
\AxiomC{$R_{\Box}w_{1}w_{2}, R_{\Box}w_{1}w_{3}, R_{\Box}w_{1}w_{4}, R_{B}w_{4}w_{6}, R_{B}w_{3}w_{6}, w_{2}: \langle A \rangle^x \overline{\phi}, w_{3}: \langle B \rangle^x \overline{\psi}, ....\ \ $
$ w_{6}: \psi, w_{6} : \overline{\psi}$}

\UnaryInfC{$R_{\Box}w_{1}w_{2}, R_{\Box}w_{1}w_{3}, R_{\Box}w_{1}w_{4}, R_{B}w_{4}w_{6}, R_{B}w_{3}w_{6}, w_{2}: \langle A \rangle^x \overline{\phi}, w_{3}: \langle B \rangle^x \overline{\psi}, ...\ \ $
$ w_{6}: \psi$}
\RightLabel{($\mathsf{IOA-U_{2}})$}
\UnaryInfC{$R_{\Box}w_{1}w_{2}, R_{\Box}w_{1}w_{3}, R_{\Box}w_{1}w_{4}, R_{B}w_{4}w_{6}, w_{2}: \langle A \rangle^x \overline{\phi}, w_{3}: \langle B \rangle^x \overline{\psi}, ...\ \ $
$ w_{6}: \psi$}
\UnaryInfC{$R_{\Box}w_{1}w_{2}, R_{\Box}w_{1}w_{3}, R_{\Box}w_{1}w_{4}, w_{2}: \langle A \rangle^x \overline{\phi}, w_{3}: \langle B \rangle^x \overline{\psi}, ...\ \ $
$ w_{4}: [B]^x\psi$}
\UnaryInfC{$D_2$}
\DisplayProof
\end{tabular}

\ \\

\begin{tabular}{@{\hskip -4em} c}

\AxiomC{$D_1\quad\quad\quad$}
\AxiomC{$\quad\quad\quad D_2$}

\BinaryInfC{$R_{\Box}w_{1}w_{2}, R_{\Box}w_{1}w_{3}, R_{\Box}w_{1}w_{4}, w_{2}: \langle A \rangle^x \overline{\phi}, w_{3}: \langle B \rangle^x \overline{\psi}, w_{1}: \Diamond ([A]^x\phi\land [B]^x\psi), w_{4}: [A]^x\phi\land [B]^x\psi$}
\UnaryInfC{$R_{\Box}w_{1}w_{2}, R_{\Box}w_{1}w_{3}, R_{\Box}w_{1}w_{4}, w_{2}: \langle A \rangle^x \overline{\phi}, w_{3}: \langle B \rangle^x \overline{\psi}, w_{1}: \Diamond ([A]^x\phi\land [B]^x\psi)$}
\RightLabel{($\mathsf{IOA-E}$)}
\UnaryInfC{$R_{\Box}w_{1}w_{2}, R_{\Box}w_{1}w_{3}, w_{2}: \langle A \rangle^x \overline{\phi}, w_{3}: \langle B \rangle^x \overline{\psi}, w_{1}: \Diamond ([A]^x\phi\land [B]^x\psi)$}
\UnaryInfC{$w_{1}: \Box \langle A \rangle^x \overline{\phi}, w_{1}:\Box \langle B \rangle^x \overline{\psi}, w_{1}: \Diamond ([A]^x\phi\land [B]^x\psi)$}
\UnaryInfC{$w_{1}: \Box \langle A \rangle^x \overline{\phi} \lor \Box \langle B \rangle^x \overline{\psi} \lor \Diamond ([A]^x\phi\land [B]^x\psi)$}
\DisplayProof
\end{tabular}
\end{small}


%
%
 \bibliographystyle{splncs04}
%

\end{document}